\documentclass[11pt]{article}   %
\usepackage{fullpage}
\usepackage[font={sf}]{caption} %

\usepackage{epsfig}

\usepackage[linesnumbered,ruled,vlined]{algorithm2e}
\usepackage{color}
\usepackage{comment}

\usepackage[backend=biber,maxbibnames=99,doi=false,isbn=false,url=false,sorting=none]{biblatex}
\providecommand{\keywords}[1]{\textbf{\textit{Keywords---}} #1}

\def\A{{\cal A}}
\def\C{{\cal C}}
\def\E{{\cal E}}

\def\M{{\cal M}}
\def\R{{\mathbb{R}}}
\def\S{{\cal S}}
\def\TF{{\cal T}}

\newcommand{\hibar}{\bar{h}_i}
\newcommand{\hjbar}{\bar{h}_j}

\RequirePackage{theorem,graphicx,amssymb,amsmath,fancyhdr,url}

\theorembodyfont{\slshape}

\newtheorem{theorem}{Theorem}
\newtheorem{lemma}[theorem]{Lemma}
\newtheorem{cor}[theorem]{Corollary}
\newtheorem{prop}[theorem]{Proposition}

\newtheorem{obs}{Observation}

\newtheorem{defini}{Definition}

\def\QED{\ensuremath{{\square}}}
\def\markatright#1{\leavevmode\unskip\nobreak\quad\hspace*{\fill}{#1}}
\newenvironment{proof}
{\begin{trivlist}\item[\hskip\labelsep{\bf Proof.}]}
 {\markatright{\QED}\end{trivlist}}

\title{Optimal Algorithms for Separating a Polyhedron from Its Single-Part Mold\thanks{This work has been supported in part by the Israel Science Foundation (grant no.~1736/19),
by NSF/US-Israel-BSF (grant no.~2019754),
by the Israel Ministry of Science and Technology (grant no.~103129),
by the Blavatnik Computer Science Research Fund,
by the Yandex Machine Learning Initiative at Tel Aviv University, and by the Natural Sciences and Engineering Research Council of Canada.}}

\author{Prosenjit Bose\thanks{School of Computer Science,
        Carleton University, {\tt jit@scs.carleton.ca}
}
\and

Tzvika Geft$^{\ddagger}$
\and
Dan Halperin\thanks{The Blavatnik School of Computer Science,
        Tel Aviv University, {\tt zvigreg@mail.tau.ac.il}, {\tt danha@post.tau.ac.il}, {\tt shasha94@gmail.com}.}
        \and
       Shahar Shamai$^{\ddagger}$
       }
\date{}

\begin{document}
\maketitle

\boldmath %
\begin{abstract}

Casting is a manufacturing process where liquid material is poured into a mold having the shape of a desired product. After the material solidifies, the product is removed from the mold. We study the case where the mold is made of a single part and the object to be produced is a three-dimensional polyhedron. Objects that can be produced this way are called castable with a single-part mold. A direction in which the object can be removed without breaking the mold is called a valid removal direction. We give an $O(n)$-time algorithm that decides whether a given polyhedron with $n$ facets is castable with a single-part mold. When possible, our algorithm provides an orientation of the polyhedron in the mold and a direction in which the product can be removed without breaking the mold.
Moreover, we provide an optimal $\Theta(n \log n)$-time algorithm to compute all valid removal directions for polyhdera that are castable with a single-part mold.
Both algorithms are an improvement by a linear factor over the previously best known algorithms for both of these problems.

\end{abstract}
\unboldmath

\keywords{Casting, Cast removal, Cast design, Separability}

\section{Introduction}

Casting is a widely-used manufacturing process, where liquid material is poured into a cavity inside a mold, which has the shape of a desired product. After the material solidifies, the product is taken out of the mold. Typically, a mold is used to manufacture numerous copies of a product, thus to ensure that a mold can be re-used the solidified product must be separated from its mold without breaking the object or the mold.

The problems that we study belong{} to the larger topic termed  \emph{Movable Separability of Sets}; see Toussaint~\cite{toussaint1985movable}. 
Problems in this area are often challenging from a combinatorial- and computational-geometry point of view (see, e.g., \cite{DBLP:journals/dcg/SnoeyinkS94}).
At the same time, solutions to these problems are needed in various application areas such as mold design~\cite{DBLP:journals/cad/AhnBBCHMS02}, assembly planning~\cite{DBLP:journals/algorithmica/HalperinLW00}, and 3D printing to mention a few.

In this paper we focus on a fairly basic movable-separability question. We are given a polyhedron $P$ in $\R^3$ with $n$ facets. No particular assumptions about the polyhedron are made beyond that it is a closed regular set, namely it does not have dangling edges or facets. The mold is box-shaped and the cavity has the shape of $P$ such that one of $P$'s facets is the top facet of the cavity. See 
Figure~\ref{fig:cast-examples} for an illustration in 2D. Once the top facet has been determined, we detect whether there is a direction in which the solidified object can be removed from the mold without colliding with the mold. Such a direction is a \emph{valid} removal direction and the corresponding top facet is \emph{valid}.
A polyhedron $P$ is \emph{castable with a single-part mold} if $P$ has at least one facet that can serve as a valid top facet.

\begin{figure}[htb]
\centering
\includegraphics[width=0.5\columnwidth]{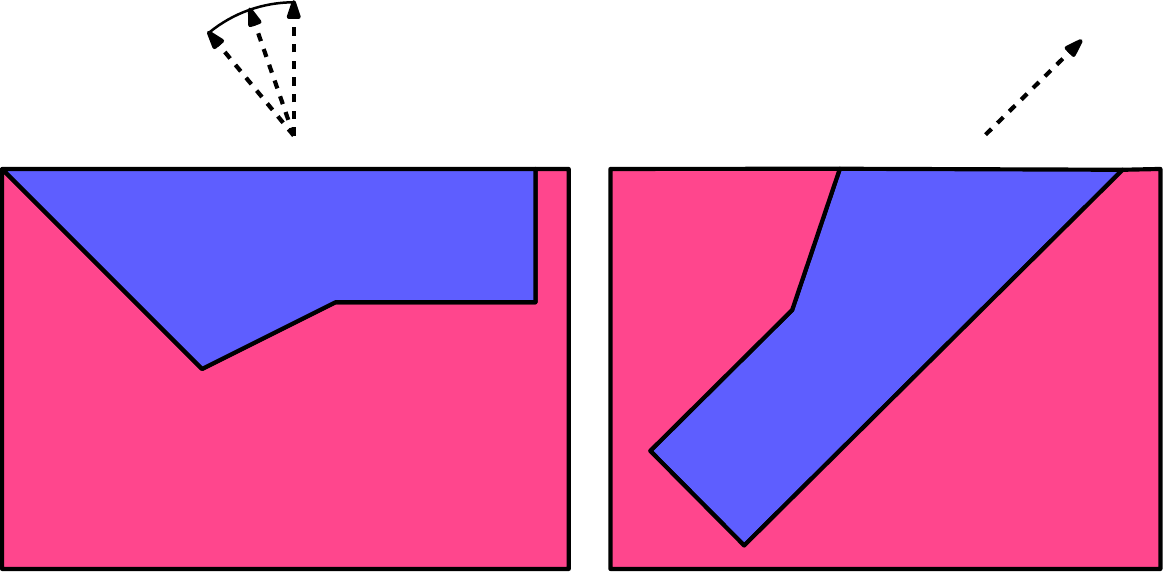}
\caption{
  Polygons (blue) in their molds (pink) and valid removal directions.
}
\label{fig:cast-examples}

\end{figure}

We address two problems:
\begin{description}

\item[All Facets Single Direction (AllFSD):] \ \\ 
Determine which facets of $P$ can serve as a valid top facet and for each such facet indicate \emph{one} valid removal direction.
\item[All Facets All Directions (AllFAD):]  \ \\ 
Same as above but for each  valid facet indicate \emph{all} the valid removal directions.
\end{description}

Why would anyone bother to solve AllFAD and not be satisfied with AllFSD? First, a solution is more stable if there is a continuum of directions rather than a single direction of separation. Second, we can use the availability of many possible directions to optimize other criteria.\\

\noindent
{\bf Previous results}
The previous best algorithms that we are aware of are based on \cite{DBLP:journals/algorithmica/AsbergBBGOTWZ97, DBLP:journals/cad/AhnBBCHMS02} and solve AllFAD in $O(n^2\log n)$ time and AllFSD in $O(n^2)$ time.
These algorithms are summarized in~\cite[Chapter~4]{Berg:2008:CGA:1370949}, which we refer to from now on.
They solve the AllFAD and AllFSD problems by solving the following two simpler problems $n$ times --- handling each candidate facet to be a top facet separately:
\begin{description}
\item[Single Facets Single Direction (SingleFSD):] \ \\ 
Determine whether a given facet $F_i$ of a polyhedron $P$ can serve as a valid top facet and if so, indicate \emph{one} direction in which $P$ can be removed from the mold with $F_i$ as its top facet.
\item[Single Facets All Directions (SingleFAD):]  \ \\ 
Same as above except indicate \emph{all} the directions in which $P$ can be removed from the mold with $F_i$ as its top facet. 
\end{description}
The existing algorithm for SingleFSD takes $O(n)$ time and the existing algorithm for SingleFAD takes $O(n\log n)$ time.
Determining whether the outer normal of a given top facet is a valid removal direction can be done in linear time~\cite{DBLP:journals/algorithmica/AsbergBBGOTWZ97}.
All the algorithms including the ones that we present use linear storage space.
Related work on computational aspects of manufacturing processes has been done for casting with a two-part mold~\cite{DBLP:journals/cad/AhnBBCHMS02, DBLP:journals/algorithmica/BoseBK97}, gravity casting, and stereolithography~\cite{DBLP:journals/algorithmica/AsbergBBGOTWZ97}; see also the survey~\cite{DBLP:journals/cg/BoseT94}.
\\

\noindent
{\bf Contribution}
Our contribution in this paper is an $O(n)$-time algorithm for the AllFSD problem and an $O(n\log n)$-time algorithm for the AllFAD problem.
Both our algorithms are optimal.
We also present an $O(n)$-time solution to the AllFAD problem when the input polyhedron is convex.
Additionally, we prove that for any polyhedron there are at most six valid top facets.
The efficiency of our solution makes it a good candidate for implementation,
which will enable product designers and engineers to quickly verify the castability of their design.

A previous version of this paper~\cite{DBLP:conf/case/BoseHS17} was presented at the IEEE 13th International Conference on Automation Science and Engineering (CASE 2017).
In this paper we extend our previous work by presenting an $\Omega(n \log n)$ lower bound for the AllFAD problem, proving that our $O(n \log n)$-time algorithm is optimal.

\section{Preliminary Analysis}
\label{sect:Preliminary-Analysis}

Instead of considering all of the top facets as possible candidates and running a separate algorithm for each of them, as in~\cite[Chapter~4]{Berg:2008:CGA:1370949}, we start by finding a small set of up to 12 candidate top facets. Then we solve SingleFSD/SingleFAD~\cite[Chapter~4]{Berg:2008:CGA:1370949} for each of the candidates.

In order to find the possible top facets we consider an arrangement of great circles on the unit sphere $\S^2$.  An arrangement of curves on the sphere is a subdivision of the sphere into vertices, edges, and faces as induced by the given curves: Vertices are the intersection points of the curves, edges are the maximal portions of a curve not intersected by any other curve, and faces are the maximal portions of the sphere that are not intersected by any curve; see, e.g., \cite{DBLP:journals/mics/BerberichFHKS10,DBLP:journals/mics/BerberichFHMW10}.  Each point $p$ on $\S^2$ represents a direction in $\R^3$---the direction of the vector from the center of $\S^2$ to $p$. We will use the terms \emph{points} and \emph{directions} on $\S^2$ interchangeably.

Let $F_1,\ldots,F_n$ be the facets of the given polyhedron~$P$. Let $\nu(F_i)$ be the normal to the facet $F_i$ pointing into the polyhedron. 
 
We describe an orientation of a mold by the pair
$(F_i, \vec{d})$, which should be interpreted as follows. The top facet
of the mold is $F_i$. To achieve this, the polyhedron needs to
be rotated by a rotation matrix $R_i$ such that $F_i$ becomes the top facet. 
We apply the same rotation matrix $R_i$ to  $\vec{d}$ to obtain a removal direction
$\vec{d_0}:=R_i\vec{d}$. (In other words, for convenience during our analysis, we use 
the vector $\vec{d}$ relative to the given, original
orientation of the polyhedron. Once we determine that the pair $(F_i, \vec{d})$ is valid we still need to rotate our polyhedron such that $F_i$ becomes the top facet and $\vec{d_0}$ is a valid removal direction in the new orientation.)

\begin{obs}
\label{obs:valid-pair2}
The pair $(F_i,\vec{d})$ is a valid orientation of a mold and removal direction if and only if for each point $p$ in the polyhedron, the ray that starts at $p$ with direction $\vec{d}$:

(i) intersects $F_i$, and 

(ii) $\forall j\neq i,$ does not intersects $F_j$. (It may partialy overlap\footnote{This corresponds to allowing $P$ to be moved out while sliding in contact with $F_j$. } $F_j$.)   
\end{obs}

\begin{lemma}
\label{lem:valid-pair}
The pair $(F_i,\vec{d})$ represents a valid orientation of a mold and removal direction if and only if

(i) $\vec{d}\cdot\nu(F_i)<0$, and  

(ii) $\forall j\neq i,\;\;  
\vec{d}\cdot\nu(F_j)\geq 0$
\end{lemma}
\begin{proof}
For the case where $F_i$ is the top facet of $P$, this fact is proved in Lemma~4.1 in \cite{Berg:2008:CGA:1370949}. It remains to notice that conditions~(i) and~(ii) of Observation~\ref{obs:valid-pair2} are invariant under rotation. They hold in $P$'s given orientation if and only if they hold when $P$ is rotated such that $F_i$ becomes the top facet.
\end{proof}

\begin{defini}
\label{def:valid-pair}

A \emph{valid pair} is a pair $(F_i,\vec{d})$ that obeys the conditions of Lemma~\ref{lem:valid-pair}.
\end{defini} 
\begin{defini}
A facet $F_i$ will be called a \emph{valid top facet} if there exists a vector $\vec{d}$ for which the pair $(F_i,\vec{d})$ is a \emph{valid pair}.
\end{defini} 

Denote by $h_i$ := $h(F_i)$ the closed hemisphere of directions $\vec{d}$ on $\S^2$ for which $\vec{d}\cdot\nu(F_i)\geq 0$, and by $\bar{h_i}$ or $\bar{h}(F_i)$ the open complement hemisphere. For a set of facets X we denote by $\bar{H}(X)$ the set $\{\bar{h}(F_i)|F_i \in X\}$. Let $\bar{H}=\bar{H}(\{F_1,\ldots,F_n\})=\{\bar{h}_1,\ldots,\bar{h}_n\}$. Let $c_i$ denote the boundary great circle of $h_i$, and let $\C=\{c_1,c_2,\ldots,c_n\}$. Consider the arrangement $\A(\C)$ on $\S^2$, namely the subdivision of $\S^2$ induced by the great circles in $\C$.
See Figure~\ref{fig:PolyToSphere} for an illustration.
\begin{figure}[htb]
\centering

\includegraphics[width=0.45\columnwidth]{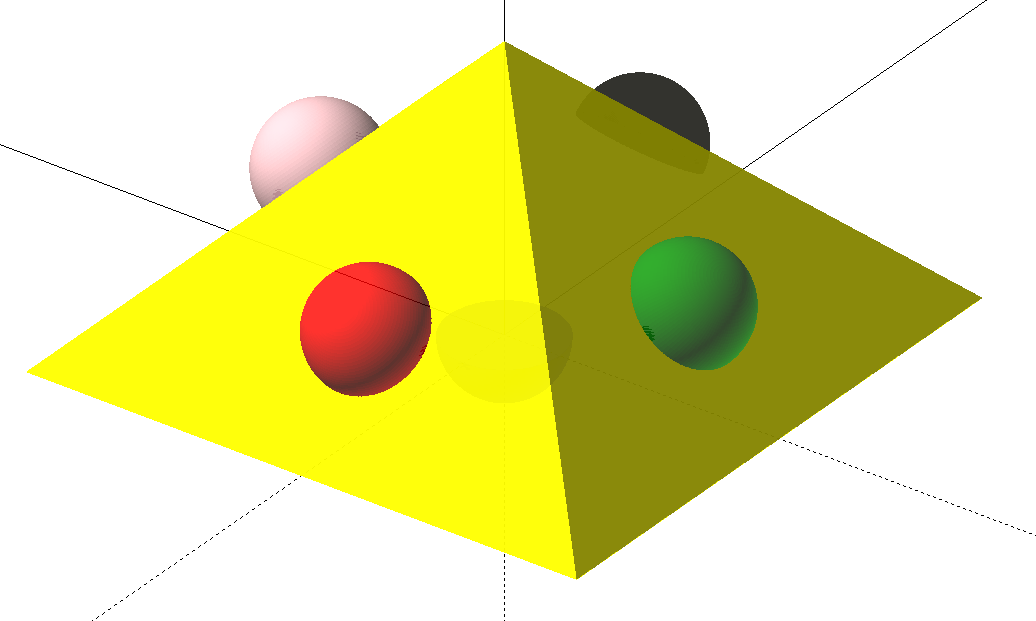}
\hspace{0.08\columnwidth}
\includegraphics[height=2.5cm]{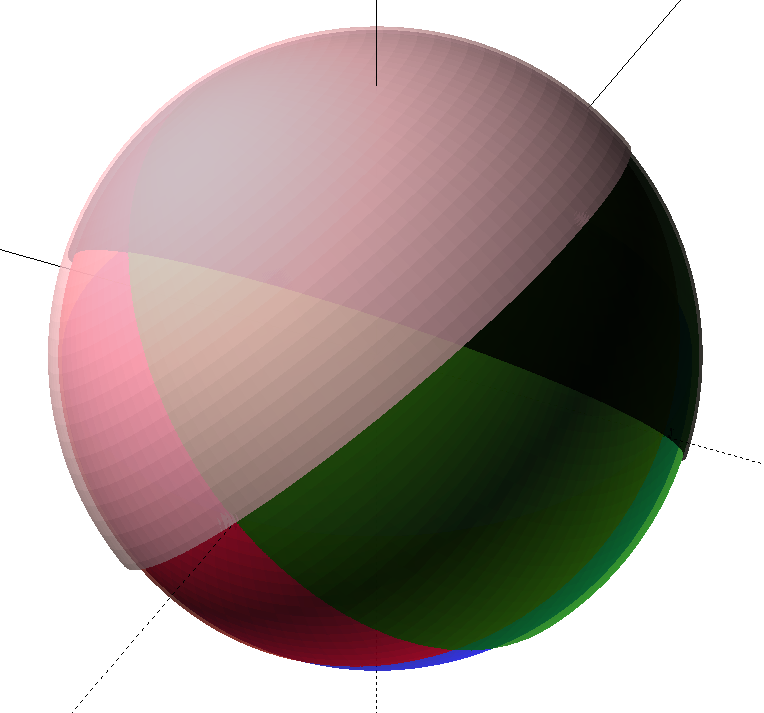}
\caption{ %
A pyramid with a hemisphere of directions $\bar{h}(F_i)$ on each of its facets (left-hand side). And $\bar{H}$, the placement of these hemisphere on $\S^2$ (right-hand side).
}
\label{fig:PolyToSphere}

\end{figure}
\begin{defini}
\label{defi:depth}

The \emph{depth} of a point $p$ on $\S^2$ is the number of hemispheres in $\bar{H}$ containing $p$.
\end{defini}    
 
\begin{obs}
All the points in any fixed cell of the arrangement $\A(\C)$ have the same depth.
\end{obs}
\begin{theorem}
\label{theo:zero-cells}
There cannot be a face of zero depth in $\A(\C)$.
\end{theorem}
\begin{proof}
Assume, for the sake of contradiction, that there exists a face of zero depth in the arrangement $\A(\C)$.
Let $\vec{d}$ be a direction in this face.
If we choose an arbitrary point inside the polyhedron (not on its boundary) and go from it in the direction $\vec{d}$ we will eventually leave the polyhedron through some facet $F_j$.
That is, $\vec{d}$ must be pointing out of the polyhedron from $F_j$, an hence $\vec{d}$ is in $\bar{h}(F_j)$, in contradiction to $\vec{d}$ lying inside a face with zero depth.

Notice that in order to avoid special degenerate cases, like crossing the polyhedron in an edge, we choose the arbitrary point such that it does not lie on any of the planes spanned by the direction $\vec{d}$ and an edge of the polyhedron.

\end{proof}

\begin{theorem}
\label{theo:valid-facet}
The polyhedron $P$ is castable with a single-part mold if and only if the arrangement $\A(\C)$ contains a point of depth~$1$.
A cell $\xi$ of depth~$1$ in $\A(\C)$, which is covered by the hemisphere
$\hibar$, represents a mold whose top facet is $F_i$ and each point $\vec{d}$ in $\xi$ represents a valid removal direction $R_i\vec{d}$, where $R_i$ is the orthonormal matrix that rotates $\nu(F_i)$ to point vertically down (in the negative $z$ direction). 
\end{theorem}
\begin{proof}
Let $\xi$ be a cell of depth~$1$ in $\A(\C)$ covered by $\hibar$, and let $\vec{d}$ be a point in $\xi$. We establish that $(F_i,\vec{d})$ is a valid pair by verifying that the conditions of Lemma~\ref{lem:valid-pair} above hold for it.

It remains to show that no point in any cell of different depth can represent a valid removal direction for any top facet. Consider a cell $\psi$ of depth greater than~$1$ and a point $\vec{d}$ in it. Let $J$ be the index set of the hemispheres $\hibar$ that cover $\psi$: 
$J=\{i|\vec{d}\in\hibar\}$. One of the facets $F_j$, for $j\in J$, must serve as the top facet for Condition~(i) of Lemma~\ref{lem:valid-pair} to hold. But then for each of the remaining facets $F_k$, for $k\in J, k\neq j$, Condition~(ii) of the lemma is violated. Finally, as shown in Theorem~\ref{theo:zero-cells}, no cell can have depth zero.
\end{proof}

\medskip

Let $\TF\subseteq\{F_1,\ldots,F_n\}$ be the set of all \emph{valid top facets}.
Our goal is to find $\TF$. %
In Section~\ref{sect:Finding A Covering-Set} we give an algorithm that finds a set of up to 12 facets that contains $\TF$ and show that $|\TF| \leq 6$.

\begin{defini}
A \emph{covering set} is a set of open hemispheres $S\subseteq \bar{H}$ such that the union of all the hemispheres in $S$ covers the entire unit sphere.
\end{defini} 

\begin{theorem}
\label{theo:covering set-top-facets}

For each covering set $S$, $\bar{H}(\TF)\subseteq S$.
\end{theorem}
\begin{proof}
Let $S \subseteq \bar{H}$ be a covering set. Assume, for the sake of contradiction that there exists some facet $F_i \in \TF$ for which $\hibar \notin S$. By Theorem~\ref{theo:valid-facet}, a facet $F_i$ is a \emph{valid top facet} if and only if there exists a depth~$1$ cell $\xi$ in $\A(\C)$ which is covered by $\hibar$ (and only by $\hibar$). We know that each point in the unit sphere is covered by some hemisphere $h \in S$, therefore some hemisphere $\hjbar \in S$ covers a point in $\xi$. By the definition of a cell, if some point in the cell is contained in $\hjbar$ then the entire cell is contained in $\hjbar$. We also assumed that $\hibar$ covers $\xi$, this means that $\xi$ is of depth at least two, in contradiction to $F_i$ being a valid top facet.
\end{proof}

Our next step is to give an algorithm that finds, in linear time, a set of 12 open hemispheres in $\bar{H}$  whose union covers the entire unit sphere.

\section{Finding A Covering-Set}
\label{sect:Finding A Covering-Set}

In order to find a \emph{covering set} for $\S^2$, we start by finding a covering set for the upper open hemisphere.
In order to do so, we will define and prove the following claims.

\begin{defini}
\label{defi:half-plane}

The half-plane of a hemisphere $h$ with respect to the upper hemisphere is the central projection of the intersection of $h$ and the upper hemisphere onto the plane $z = 1$~\cite[Chapter~4]{Berg:2008:CGA:1370949}. See Figure~\ref{fig:proj}.
\end{defini}
\begin{figure}[htb]
\centering
\includegraphics[width=0.45\columnwidth]
{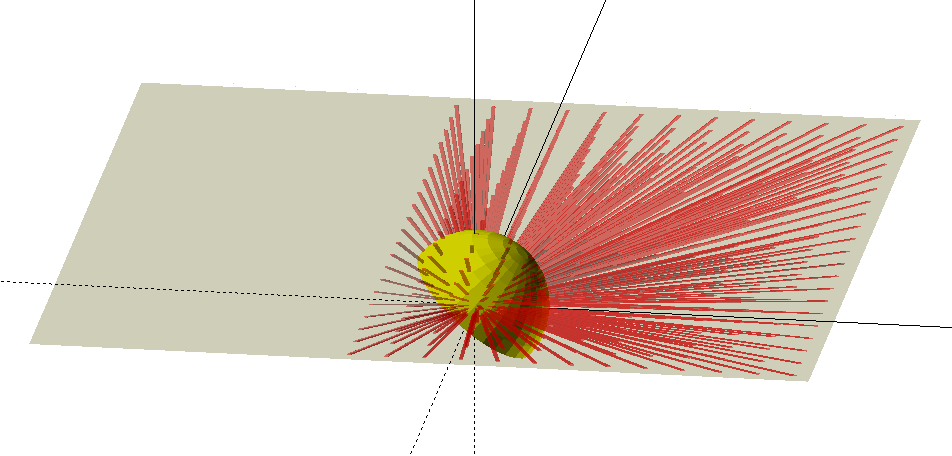}
\includegraphics[width=0.45\columnwidth]
{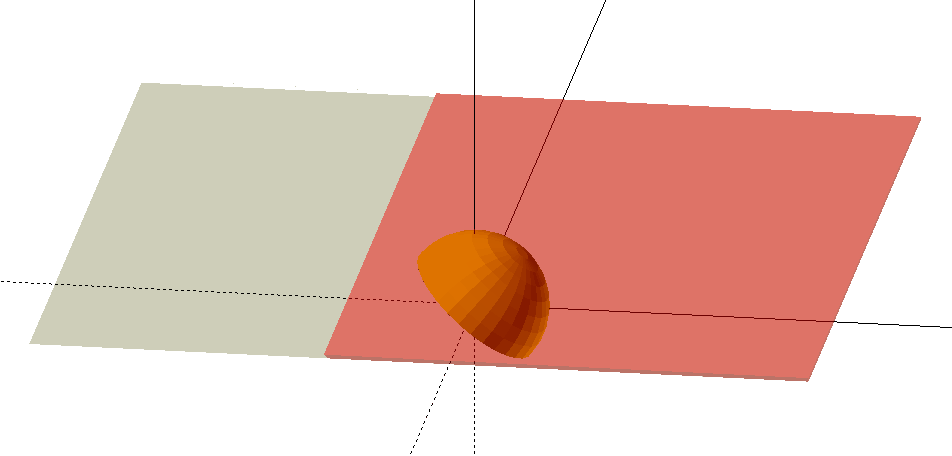}
\caption{A visualization of a central projection of an arbitrary hemisphere (yellow) onto the plane $z = 1$ where it becomes a half-plane (red).
For clarity of illustration we show a sample of the projection on the left and the full half-plane on the right.
 \label{fig:proj}
}
\end{figure}

\noindent
{\bf Remark.}
The ``half-plane'' of the lower hemisphere  with respect to the upper hemisphere is empty and the ``half-plane'' of the upper hemisphere with respect to the upper hemisphere is the entire plane $z = 1$.

\begin{prop}
\label{lem:helly}
If the union of a set of half-planes, $B$,  is the entire plane then there exists a set $S \subseteq B$ such that the union of the half-planes in $S$ is the entire plane and $|S| = 3$.
The set $S$ can be computed in linear time in the size of $B$.
\end{prop}
\begin{proof}
Helly's theorem \cite{danzer1921helly} states that given a collection $X_1,X_2,...,X_n$ of convex subsets of ${\R}^d$ where $n>d$, if the intersection of any $d+1$  objects in this collection is nonempty, then the entire collection has a nonempty intersection.
The contrapositive of this theorem is that if the intersection of the entire collection is empty then there exists a subset of size $d+1$ such that its intersection is empty.\\
In our case $d=2$, \\
$\bigcup\limits_{b \in B} b = {\R}^2 \Rightarrow	\overline{\bigcup\limits_{b \in B} b } =   \emptyset \Rightarrow	 \bigcap\limits_{b \in B} \overline{b} = \emptyset$.
By Helly's theorem there exist $ \overline{b_1} \cap \overline{b_2} \cap \overline{b_3}  = \emptyset$ in $B$.
By that we learn that  $b_1 \cup b_2 \cup b_3 = {\R}^2$.
This subset can be computed in linear time in the size of $B$ with linear space using a version \cite{norton1992using} of Megiddo's LP algorithm  \cite{Megiddo:1984:LPL:2422.322418}, or in linear expected time and constant space using Seidel's randomized incremental algorithm  \cite{Seidel1991}. 
\end{proof}

We can now find a covering set  $S \subseteq \bar{H}$ of size three for
the upper open hemisphere in linear time, as follows. If
some hemisphere in $\bar{H}$ is the upper open hemisphere return it and
stop; we are done. Otherwise, transform each
hemisphere in $\bar{H}$ into a half-plane on $z = 1$, as described
earlier. Then, use the procedure we have just described to find
three half-planes that cover the entire plane $z = 1$ in linear
time.
A covering set for any open hemisphere in $\S^2$ can be
found similarly.
In summary,

\begin{obs}
\label{obs:hemisphere-covering}
Any open hemisphere in $\S^2$ is covered by up to three hemispheres in $\bar{H}$.
\end{obs}

In order to find a covering set $S \subseteq \bar{H}$ for $S^2$, we choose any four open hemispheres that cover $S^2$ (e.g., take the set $\bar{H}(X)$, where $X$ is the set of facets of some tetrahedron, which covers $S^2$ by Theorem~\ref{theo:zero-cells}).
We then find a covering set of size three for each of these hemispheres, as described above.
Finally, by combining the resulting covering sets, we obtain a covering set for $S^2$ of total size 12.

We continue by proving the existence of a covering set for
$S^2$ of size six.

\pagebreak
\begin{theorem}
\label{theo:min-covering}
There exists a covering set of size six for $S^2$.

\end{theorem}
\begin{proof}
Let $G$ be some great circle on $\S^2$ that does not contain a vertex of $\A(\C)$ and is not contained in $\C$.
$G$ is composed of parts of cells in $\A(\C)$ that are partially contained in the two open hemispheres defined by $G$. This means that any covering set of an open hemisphere defined by $G$, covers $G$ as well.

By that and by Observation~\ref{obs:hemisphere-covering}, the covering set of $\S^2$ is of size six at most.
\end{proof}

\noindent
{\bf Remark.} This is only a proof of
existence.
We do not provide an algorithm to compute a covering set of size six.
The problem with turning the proof of Theorem~\ref{theo:min-covering} into a linear-time algorithm is that we do not know how to find deterministically in linear time a great circle on $S^2$ that avoids all the vertices of the arrangement.                                          

\begin{theorem}
\label{theo:six-top-facets-1}
The number of valid top facets for any polyhedron is at most six and this bound is tight.
\end{theorem}

\begin{proof}
As we showed in Theorem~\ref{theo:min-covering}, there exists a covering set of size six, and by Theorem~\ref{theo:covering set-top-facets}, we know that the number of valid top facets in a polyhedron is bounded by the size of any covering set.
Thus, the number of valid top facets is bounded by six.

Notice that this is tight---all the six facets of a parallelepiped are valid top facets. 
\end{proof}

\section{Algorithms}
The algorithms for both the AllFAD and AllFSD problems are similar. 
We first find a covering set of constant size, as previously discussed. Then, for each facet in the covering set we solve SingleFAD or SingleFSD~\cite[Chapter~4]{Berg:2008:CGA:1370949}, depending on whether we are solving AllFAD or AllFSD. 
Finally, we return all the removal directions for each valid top facet in the covering set.

Algorithm~\ref{alg:threeD} sketches both solutions.
The difference between the AllFSD and AllFAD variants is in the function handleSingleFacet($P$, $F$), which for a given polyhedron and a facet, solves SingleFAD in $O(n\log n)$ time or SingleFSD in linear time \cite[Chapter~4]{Berg:2008:CGA:1370949}.

\begin{algorithm}
\DontPrintSemicolon
\KwIn{$P$ a polyhedron}
\KwOut{ A list of all the top facets and their removal direction(s)}
$T \gets \emptyset$\;

$C \gets$ coveringSet($P$)\;
\For{$c \in  C$}{
	\emph{directions} $\gets$ handleSingleFacet($P$, \emph{c.facet})\;
  	\uIf {directions.notEmpty}{
  		$T$.add(\emph{c.facet}, \emph{directions})\;
  	}
}
\Return{$T$}\;
\caption{3D AllFAD/AllFSD}
\label{alg:threeD}

\end{algorithm}

We obtain the following:

\begin{theorem}
Given a polyhedron $P$ as input, our algorithm for AllFSD outputs a list of all valid top facets and a single removal direction for each such facet in total $O(n)$ time, where $n=|P|$.
\end{theorem}

\begin{theorem}
Given a polyhedron $P$ as input, our algorithm for AllFAD outputs a list of all valid top facets and all valid removal direction for each such facet in total $O(n \log n)$ time, where $n=|P|$.
\end{theorem}

When solving the AllFAD variant, one may wish to run the AllFSD algorithm  first in order to end up with a linear-time algorithm in case that the polyhedron is not castable with a single-part mold.
We remark that for the 2D analog of our problem, namely when the object to be cast is a polygon, Asberg et al.~\cite{DBLP:journals/algorithmica/AsbergBBGOTWZ97} %
solve AllFAD in linear time and constant space using similar techniques.

\section{Lower Bound}
In this section we show an $\Omega (n \log n)$ lower bound for SingleFAD.
We establish the lower bound in the algebraic computation tree model~\cite{ben1983lower} using a reduction from sorting.
The lower bound proves the optimality of both the existing SingleFAD algorithm~\cite[Chapter~4]{Berg:2008:CGA:1370949} and our new algorithm for \mbox{AllFAD} in this model.

\begin{figure*}[htp]
\centering
\includegraphics[width=0.42\textwidth] %
{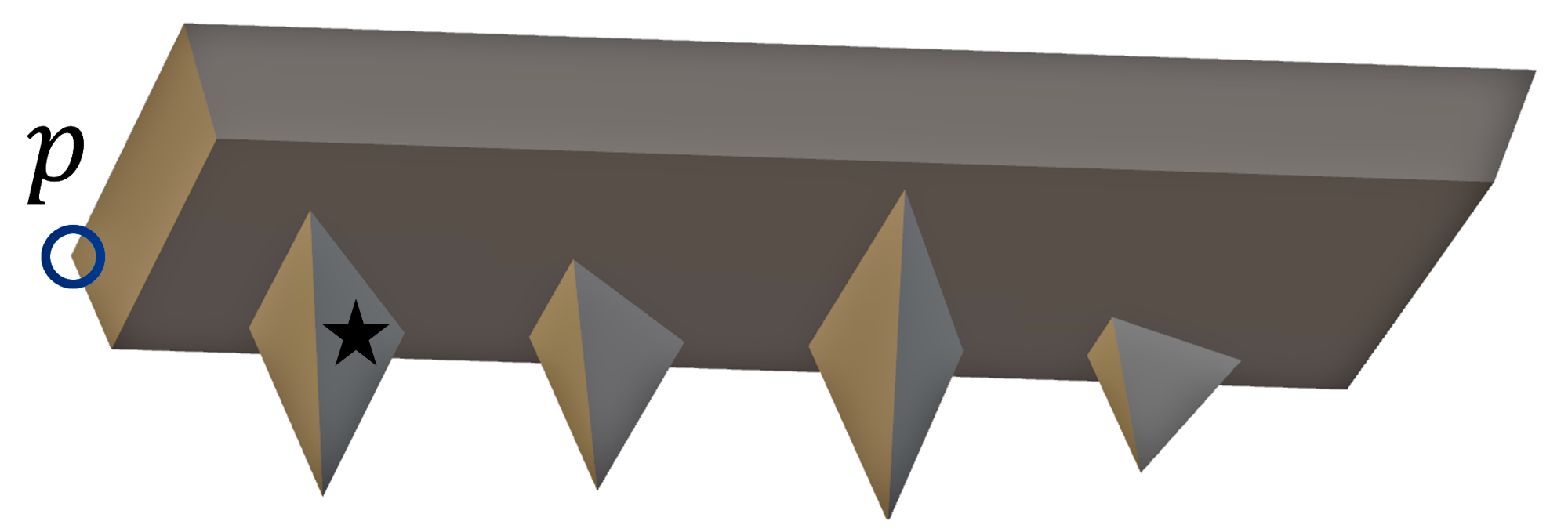}
\hfill
\includegraphics[width=0.30\textwidth] %
{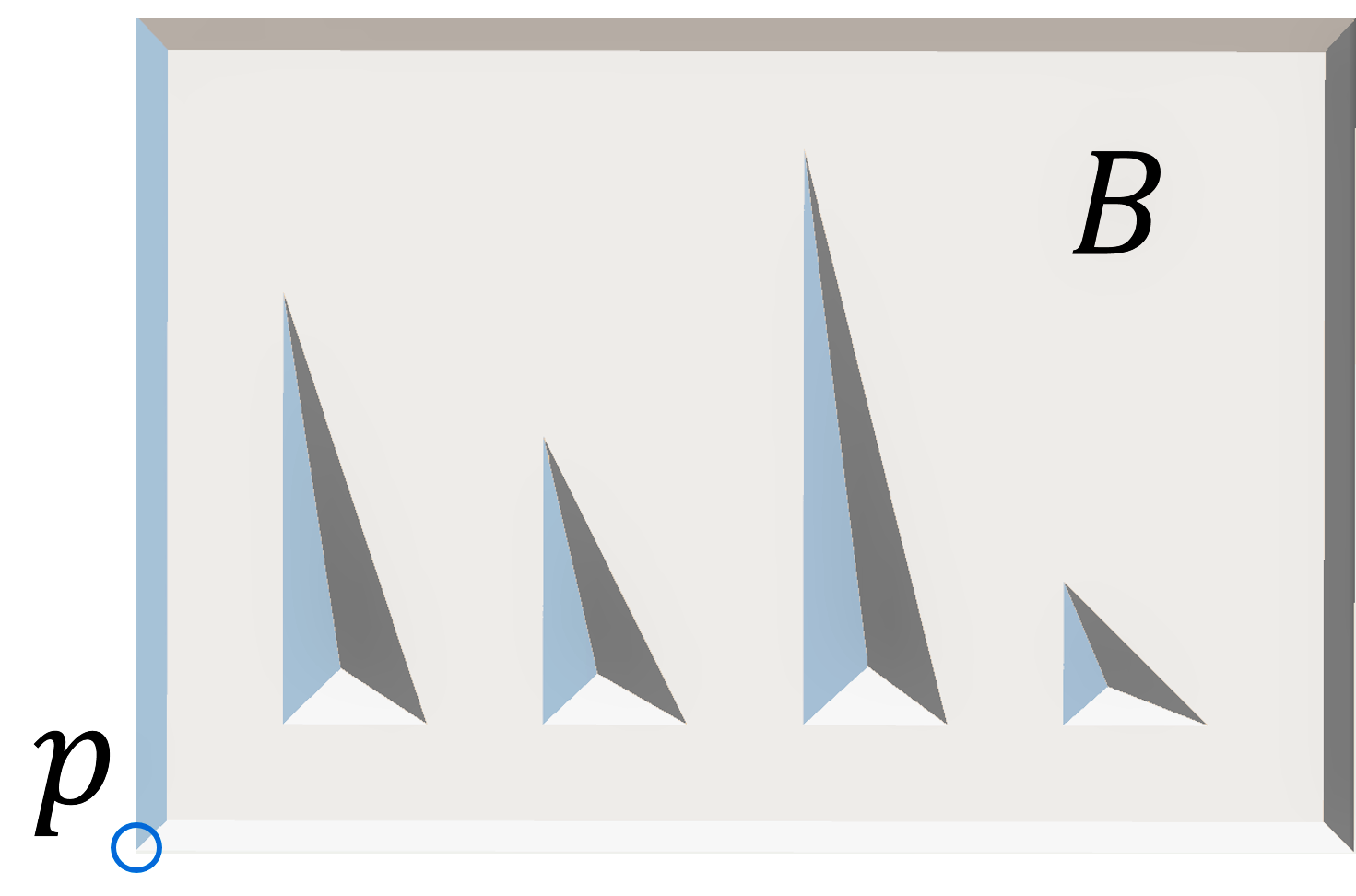}
\hfill
\includegraphics[width=0.2\textwidth] %
{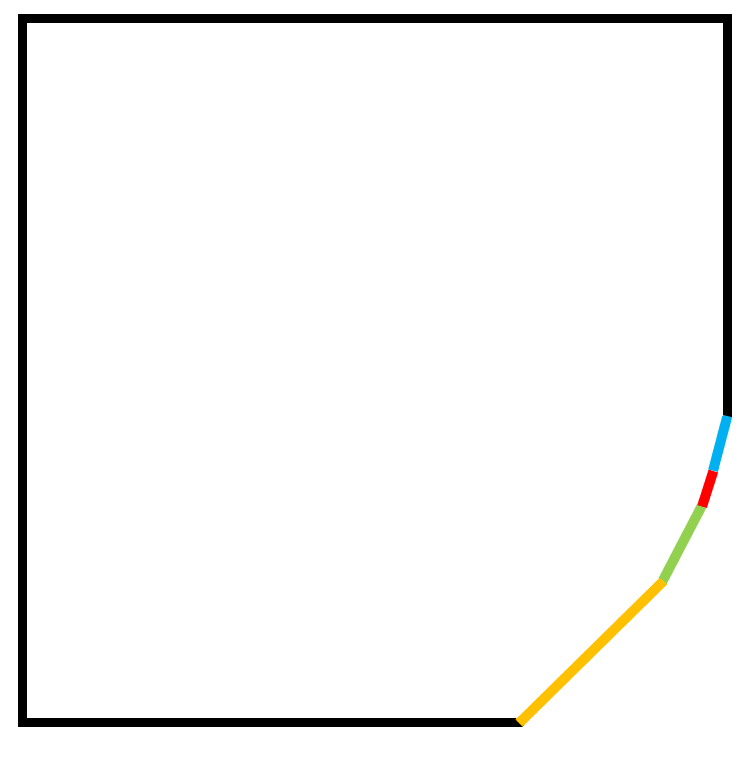}
\caption{A construction of $P$ for $X=\{3,2,4,1\}$ viewed from the side (left) and from below (middle).
The tetrahedra $T_1, \ldots, T_4$ appear from left to right and their darkest facets are the $F_i$'s (the facet marked by a star is $F_1$).
The corresponding polygon $Q$ that represents the removal directions for $P$, with the four slopes sorted in ascending order from the bottom up in color (right).
 \label{fig:reduction}
}
\end{figure*}

Let $X = \{x_1, \ldots, x_n\}$ be $n$ positive real numbers, which we wish to sort.
We describe a polyhedron $P$ such that the convex polygon $Q$ representing all the removal directions of its top facet will contain the numbers of $X$ as the slopes of its edges.
We can then read off $X$ in sorted order from the edges of $Q$.
We begin with $P$ as a rectangular frustum whose two bases, which are rectangular facets, are parallel to the plane $z=0$; see Figure~\ref{fig:reduction}.
Let $B$ denote the bottom base of $P$, which is smaller than the top base.
For a facet $F$ of $P$, let $g(F)$ denote the half-plane obtained by the central projection of $F$'s inner closed hemisphere, $h(F)$, onto the plane $z=1$ (as in\footnote{$g(F)$ is a half-plane on $z = 1$ containing all the directions $\vec{d}$ with a positive $z$ component such that $\vec{d}\cdot\nu(F)\geq 0$} Definition~\ref{defi:half-plane}) and let $\ell(F)$ denote the line that defines $g(F)$.
We call $g(F)$ the \emph{corresponding half-plane} of $F$ and $\ell(F)$  the \emph{corresponding line} of $F$. %
Other than the two bases of $P$, each facet $F$ in $P$ will have the property that $g(F)$
contains the origin and $\ell(F)$ is tangent to the unit circle.
For such a facet $F$, we can find its inner normal $\nu(F)$ such that we attain a desired line $\ell(F)$ that is tangent to the unit circle as follows:
Take the normal vector $v$ of $\ell(F)$ that points to the origin and tilt it upwards about $\ell(F)$ by $\pi/4$, which results in $\nu(F)$.

We now describe the remaining facet of $P$. The corresponding lines of the four non-base facets of $P$ are $x=-1,x=1,y=-1,y=1$, respectively.
For each $x_i \in X$, $P$ has a corresponding facet $F_i$ such that $\ell(F_i)$ has the slope $x_i$.
To add $F_i$ to $P$, we define a tetrahedron $T_i$ whose top facet fully overlaps with $B$, i.e., $T_i$ is attached to $B$ from below.
The remaining three facets of $T_i$ are $F_i$ and two facets whose corresponding lines are $x=-1$ and $y=1$, respectively.
The $T_i$'s are arranged in a row along $B$.

After applying a SingleFAD algorithm on the top base of $P$, the output convex polygon $Q$ is the common intersection of the corresponding half-planes of all of $P$'s facets except the top base (each point in the common intersection corresponds to a removal direction).
The corresponding half-planes of the facets $F_1, \ldots, F_n$ contribute to a contiguous set of edges along $Q$ that are ordered by their slopes, which are $x_1, \ldots, x_n$.
Finally, $P$ has $3n+6$ facets and can be constructed in linear time.
Therefore, we obtain the following:

\begin{theorem}
There is an $\Omega (n \log n)$ lower bound for SingleFAD and AllFAD in the algebraic computation tree model~\cite{ben1983lower}.
\end{theorem}

\begin{cor}
Our algorithms for AllFSD and AllFAD are optimal.
\end{cor}

A natural question is whether we can obtain even more efficient algorithms for special classes of polyhdera.
In the next section we show that this is indeed possible for convex polyhedra.
We now consider another restricted class of polyhedra and show that the lower bound still holds for it.

A polyhedron $\mathcal{P}$ is called \emph{star-shaped} if it contains a point $p$ such that, for any other point $q$ in $\mathcal{P}$, the line segment $\overline{pq}$ lies in $\mathcal{P}$, i.e., $p$ can see all of $\mathcal{P}$.
Such a point $p$ exists if the intersection of the half-spaces defined by the facets of $\mathcal{P}$ (facing into $\mathcal{P}$) is not empty.
We may modify the construction of $P$ above so that $P$ is star-shaped.
For the point $p$ that can see of all $P$, we take the top left vertex (in the $xy$-plane) of the top facet, which is parallel to $B$; see Figure~\ref{fig:reduction}.
The point $p$ is contained in the half-spaces corresponding to all facets of $P$ except for facets of the tetrahedra $T_1, \ldots, T_n$ that have the corresponding lines $x=-1$ and $y=1$.
We may tilt such facets so that their corresponding lines become $x=-c$ and $y=c$ respectively, where $c$ is sufficiently large so that $p$ can fully see the tilted facets.
This change has no effect on the possible removal directions, since the frustum part of $P$ remains the same, and so it still contains the two facets with the corresponding lines $x=-1$ and $y=1$.
Therefore, $Q$ remains the same and we may read off $X$ in sorted order from its edges as before.

\section{Casting Convex Polyhedra}
In this section we show how to solve the AllFAD problem for a \emph{convex} polyhedron more efficiently than for an arbitrary polyhedron, namely in $O(n)$ time.
We start with some simple observations.
\begin{obs}
\label{obs:consecutive-normals}
Let $Q$ be a convex polygon in the plane whose edges are given in cyclic order $\E$, say clockwise. Given a direction $\vec{d}$ in the plane, all the edges of $Q$ whose inner-facing normals have a non-negative scalar product with $\vec{d}$ form a consecutive subchain of $\E$.
\end{obs}
Let $\pi$ be a plane in $R^3$ that is parallel to a direction $\vec{d}$, intersects some convex polyhedron, $P$, but does not intersect any of its vertices. Let $Q$ be the convex polygon, which is the non-empty intersection of the convex polyhedron $P$ with $\pi$. For each $F_i$ that intersects with $\pi$ we denote by $e_i$ the edge $F_i \cap \pi$ of $Q$. Let $\nu(e_i)$ denote the inner facing normal of $e_i$ in $Q$.
\pagebreak
\begin{lemma}
\label{lem:normal-projection}
If ${\rm sign}(\nu(F_i)\cdot\vec{d})\neq 0$, then ${\rm sign}(\nu(e_i)\cdot\vec{d})= {\rm sign}(\nu(F_i)\cdot\vec{d})\;$.	
\end{lemma}
\begin{proof}
Let $D$ be a vector in direction $\vec{d}$ and length $\varepsilon$ for some small $\varepsilon$. 
${\rm sign}(\nu(F_i)\cdot\vec{d})$ indicates whether the direction $\vec{d}$ points into or out of $P$ when it starts from a point on $F_i$~\cite[Chapter~4]{Berg:2008:CGA:1370949}, i.e, if ${\rm sign}(\nu(F_i)\cdot\vec{d})>0$ (${\rm sign}(\nu(F_i)\cdot\vec{d})<0$) for each point $p \in F_i$, then $p+D$ is inside (outside of) the polyhedron.

After the intersection with $\pi$, each point on $e_i$ must still fulfill these conditions, since for each point $p \in e_i$, $p+D$ is in $\pi$ and therefore it is in the polygon $Q$ if and only if it is in the original polyhedron.

\end{proof}
We say that two facets of the input polyhedron $P$ are neighbors if their closures intersect at an edge of the polyhedron.  Denote by $\M_i$ the set of neighbors of the facet $F_i$, by $m_i$ the cardinality of this set, and by $J_i$ the index set of the facets in $\M_i$.
The efficient AllFAD algorithm for convex polyhedra stems from the observation that it suffices to consider only the neighbors of $F_i$ in order to determine if $F_i$ is valid, which we formalize next.
\begin{lemma}
\label{lem:valid-pair-convex}
For a convex polyhedron, the pair $(F_i,\vec{d})$ represents a valid orientation of a mold and removal direction if and only if 

(i) $\vec{d}\cdot\nu(F_i)<0$, and 

(ii)  $\forall j\in J_i,\;\;  
\vec{d}\cdot\nu(F_j)\geq 0$.
\end{lemma}
\begin{proof}
The only difference between Lemma~\ref{lem:valid-pair} and the current lemma is the set of facets on which Condition~(ii) is tested: here it is only on the neighboring facets of $F_i$. We argue that if the condition holds for the neighboring facets it will hold for all facets other than $F_i$.
Assume, for a contradiction that it holds for all neighboring facets, and there is a non-neighboring facet $F_k, k\neq i$ for which it does not hold. Let $s$ be a segment connecting a point inside $F_i$ and a point inside $F_k$. Let $\pi$ be a plane containing $s$ and parallel to $\vec{d}$. If $\pi$ intersects a vertex of $F_i$ then we move s (and $\pi$) slightly parallel to itself such that it does not cross any vertex of $F_i$.   Let $Q:=P\cap \pi$, and as above let $e_t:=F_t\cap\pi$ for every facet $F_t$ of $P$ that intersects $\pi$.
Let $e'$ and $e''$ be the two neighboring edges to $e_i$ in $Q$ --- notice that $e'$ and $e''$ are the intersection of $\pi$ with two neighboring facets of $F_i$.
By Lemma~\ref{lem:normal-projection} the scalar product of $\vec{d}$  with both $\nu(e')$ and $\nu(e'')$ is non-negative and with both $\nu(e_i)$ and $\nu(e_k)$ is negative. However, these form an intertwining cyclic subsequence of edges of $Q$: $e',e_i,e'',e_k$, in contradiction with Observation~\ref{obs:consecutive-normals}. This means that the inner-facing normal of $e_k$ has non-negative scalar product with $\vec{d}$ and the same holds for the inner-facing normal of $F_k$.   
\end{proof}
\begin{theorem}
\label{theo:convex-neighbors-intersection}
Given a convex polyhedron $Q$ and a specific top facet $F_i$, it is possible to find \emph{all} the removal directions in $O(m_i)$ time.
\end{theorem}
\begin{proof}
Without loss of generality let's assume that $F_i$ is  horizontal (parallel to the $xy$-plane) and $\nu(F_i)$ points vertically downwards. We wish to find all the removal directions for which the conditions of Lemma~\ref{lem:valid-pair-convex} hold.
This can be done in $O(m_i\log m_i)$ time by finding the half-plane $g(F_j)$ for each facet $F_j$ in $\M_i$ and intersecting all of the resulting half-planes (see~\cite[Chapter~4]{Berg:2008:CGA:1370949}). 
We wish to reduce the running time of this procedure to $O(m_i)$ time. Luckily, these half-planes are sorted by their slope since the slope of the half-plane created by $F_j$ on $z = 1$ is equal to the slope of the edge on $F_i$ created by $F_j$ (again, $z = 1$ and $F_i$ are parallel).
The edges of a convex polygon are ordered by their slope.
The intersection of a given set of half-planes which are sorted according to the slope of their bounding lines,  can be computed in $O(m_i)$ time~\cite{KEIL1991121}. 
\end{proof}

The algorithm proceeds by fixing a facet $F_i$, computing its edge-neighboring facets and computing the set of allowable directions using Theorem~\ref{theo:convex-neighbors-intersection}. We repeat the procedure for each facet of $P$. Notice that $m_i$ is in fact the number of edges on the boundary of $F_i$. The overall cost of $O(m_i)$ time over all candidate top facets $F_i$ is $O(n)$ time by Euler's formula \cite[Chapter~13]{aigner2010proofs}, In summary,

\begin{theorem}
Given a convex polyhedron $P$ with $n$ facets, we  can solve the AllFAD problem for $P$ in time $O(n)$.
\end{theorem}

We present one more result concerning the relation between the castability of a polyhedron $P$ and the castability of its convex hull $C\!H(P)$.
We show that every valid pair (see Definition~\ref{def:valid-pair}) of $P$, induces a corresponding valid pair for $C\!H(P)$. This means that the set of valid pairs of the convex hull of $P$ is a super-set of the valid pairs of $P$.

\begin{theorem}
\label{theo:polyhedron-convex-hull-castbility}
Given a polyhedron $P$, one of its valid top facets $F_i$, and a removal direction $\vec{d}$, $P$'s convex hull $Q$, can be removed through $CH(F_i)$ (the convex hull of $F_i$) with direction $\vec{d}$.
\end{theorem}
\begin{proof}
Assume, for the sake of contradiction, that $Q$ cannot be removed in direction $\vec{d}$ through $CH(F_i)$. This means that for some point $q \in Q$ one of the conditions of Observation~\ref{obs:valid-pair2}  does not hold.
$Q$ is convex, and therefore the open ray that emanates from $q$ in direction $\vec{d}$ cannot intersect two facets of $Q$.  This means that any point in $Q$ that fulfills the first condition must fulfill the second condition as well. Thus, the first condition must not hold for $q$. Let $v$ be a ray in direction $-\vec{d}$, and $L$ be the locus of all the points (not necessarily in $Q$) for which the first condition holds. $L=CH(F_i) \oplus v$~ (where $\oplus$ represents Minkowski sums \cite{DBLP:journals/comgeo/AgarwalFH02}).
Every point in $P$ satisfies the first condition of Observation~\ref{obs:valid-pair2}, therefore $P\subseteq L$ and since $L$ is convex, $Q\subseteq L$. Contradiction.
\end{proof}

A similar relationship between a polyhedron and its convex hull was shown for castability with a two-part mold~\cite{DBLP:journals/algorithmica/BoseBK97} and for feasibility of stereolithography~\cite{DBLP:journals/algorithmica/AsbergBBGOTWZ97}.

\section{Conclusion and Further Work}
We have described efficient solutions to determine the castability of a polyhedron with a single-part mold. Our algorithms are an order of magnitude faster than the best previously known algorithms for these problems.
We showed that our algorithms are optimal in the algebraic computation tree model.
\\\\
We outline several directions for further research: 

(i) Our focus here was on separability with a single translation. Can one devise efficient algorithms in case we are allowed to remove the polyhedron from its cast using arbitrary motion? 
This problem is already interesting for separating a planar polygon from its planar cast.
The latter can be resolved by considering the
motion planning problem for a robot (the polygon) among
obstacles (the mold). This approach however would yield a
near-quartic time solution~\cite{DBLP:journals/dcg/HalperinS96}. Thus, the goal here is to take advantage of the special structure of
the problem to obtain a more efficient solution.

(ii) There are many interesting problems when the mold is made of two or more parts (see, e.g.,~\cite{DBLP:journals/cad/AhnBBCHMS02}), and we hope that our novel observations here may open the door to more efficient algorithms for these more complicated casting problems.

\printbibliography

\end{document}